\documentclass[journal]{IEEEtran}
\usepackage{blindtext}
\usepackage{graphicx}
\usepackage{enumerate}

\ifCLASSINFOpdf
\else
\fi
\usepackage[cmex10]{amsmath}
\usepackage{amsfonts}
\usepackage{breqn}

\DeclareMathOperator*{\argmin}{arg\,min}
\DeclareMathOperator*{\argmax}{arg\,max}
\begin{document}
%
\title{On Matched Metric and Channel Problem}
%
%
%

\author{Artur~Poplawski}
\maketitle

\begin{abstract}
The sufficient condition for partial function from the cartesian square of the finite 
set to the reals to be "compatible" with some metric on this set is given.
It is then shown, that when afforementioned set and function are respectively a space of binary words 
of length $n$ and probalities of receiving some word after sending the other word through Binary Assymetric Channel
the condition is satisfied so the required metrics exist.
This proves under the slightly weaker definition of matched metric and channel conjecture stated in \cite{mmc} 
\end{abstract}


%
\IEEEpeerreviewmaketitle

\section{Introduction}
In \cite{mmc} authors consider following problem originating  
from Information Theory: for which channel models there 
is metric $d$ on the space of the codewords such that following matching 
condition holds for each codewords x, y, z:
$ Pr(x | y ) > Pr(x | z ) $ if and only if $d(x,y) < d(x,z) $
Here $P(x | y )$ means probablity of the receiving codeword $x$ assuming 
that codeword $y$ was sent.
\cite{mmc} gives extensive overview of the history, current literature and 
state of knowledge on the subject. It also proves existence of such a metric 
in case of so called Z-channels and for the codewords of length 2 and 3 
in the case of Binary Assymetric Channel (BAC).
Finally authors of \cite{mmc} state the conjecture that such a compatible 
metric exists for the space of the codewords 
of arbitrary length  $n$ for BAC.

In section II existence of  metrics compatible to certain function will be considered 
(in slightly more general case comparing to \cite{mmc}). 
In this case necessary and sufficient condition for existence 
of such a metrics will be given. In section III we will discuss relation between 
the theorem from II and information theoretic 
case and also with some of the notions introduced in \cite{mmc}. Section IV 
applies result from II to the case of BAC and proves (with some corrections discussed in III) 
conjecture from \cite{mmc}.


\section{Compatible metrics}
Let $X$ be a finite set and let $$S \subset X^{2}-\Delta$$ where $$\Delta=\{(x,x) : x \in X\}$$
Let $f$ be a function:
$$ f: S \rightarrow \mathbb{R} $$ We have follwiong theorem:
\newtheorem{theorem}{Theorem}
\newtheorem{remark}{Remark} 
\newtheorem{lemma}{Lemma} 
\newtheorem{definition}{Definition}

\begin{theorem}
If for each $n$ and sequence $x_{0}, \ldots x_{n-1}$ such that: 
$$f((x_{0}, x_{n-1})) > f((x_{0}, x_{1}))$$
$$f((x_{1}, x_{0})) > f((x_{1},x_{2}))$$ 
$$\ldots$$ 
$$f((x_{n-2},x_{n-3})) > f((x_{n-2},x_{n-1}))$$
and $(x_{n-1},x_{n-2}) \in S$ and $(x_{n-1},x_{0}) \in S$ we have: 
$$f((x_{n-1},x_{n-2})) \leq f((x_{n-1},x_{0}))$$

then there exists metric $d$ on $X$ such that for all $x,y,z$ 
$f(x,y) > f(x,z) \Rightarrow d(x,y)<d(x,z) $
\end{theorem}

\begin{proof}
Let $U$ be the set of unordered pairs of elements of $X$ so:
$$
U = \{s \subset X : |s| = 2\}
$$ 
Let's define relation $R \subset U^2$ in the following manner:
for $a, b \in S^2$ $aRb$ if one of the two conditions below holds:
$$a = b$$ or 
$$a=\{x,y\},  b=\{x,z\}, y \neq z,  f(x,y) > f(x,z)$$

$R$ is reflexive (what is obvious) and antisymmetric. Let $\overline{R}$ be the 
transitive closure of the $R$.
It is, again, reflexive and by definition transitive. We will show that it is 
antisymmetric.
Let's assume that this is not true, so we have $a$ and $b$, such that $a \neq\ b$ 
and $a\overline{R}b$ and $b\overline{R}a$ but $a \neq b$.
Since $\overline{R}$ is transitive closure of $R$ we would have, that there 
exists $n$ and $m$ and two sequences (after appropriate allignement of idices:
$$
a = z_{0}, z_{2}, \ldots, z_{n} = b
$$
and
$$
b = z_{n}, z_{n+1}, \ldots, z_{n+m-1}, z_{n+m} = a = z_{0}
$$
such that
$$
z_{0} R z_{1}, z_{1} R z_{2}, \ldots z_{n+m-1} R z_{0} 
$$ 

From the definition of $R$ and because $z_{i} \in U$ we have for $0 \leq i < n+m$ 
$z_{i}=\{x_{i}, x_{(i-1) \bmod n+m}\}$ where $x_{i} = z_{i} \cap z_{(i+1) \bmod (n+m)}$

It means, that we have respectively:
$$f((x_{0}, x_{n+m-1})) > f((x_{0}, x_{1}))$$
$$f((x_{1}, x_{0})) > f((x_{1},x_{2}))$$ 
$$\ldots$$ 
$$f((x_{n+m-2},x_{n+m-3})) > f((x_{n+m-2},x_{n+m-1}))$$

so, by the assumption on $f$ we must have:

$$f((x_{n+m-1},x_{n+m-2}) \leq f((x_{n+m-1},x_{0}))$$

at the other hand, since $z_{n+m-2} R z_{n+m-1}$ we have:

$$f((x_{n+m-1},x_{n+m-2}) > f((x_{n+m-1},x_{0}))$$
what is contradiction.

This proves, that $\overline{R}$ is a antysymmetric so a partial order in $U$.
As a partial order it can be extended to total (linear) order $\overline{\overline{R}}$, and, 
since the $X$ so also $U$ is finite there is a function $g:U\rightarrow\mathbb{R}^{+}$ 
such that $x\overline{\overline{R}}y$ iff $g(x)<g(y)$ 
Now we will use the trick from \cite{mmc} Lemma 7 (for the self–containedness of the work, 
Appendix Lemma 1 of this note repeats statement and proof of the Lemma 7 from \cite{mmc}). 
Let's define the $e:X^{2} \rightarrow \mathbb{R}^{+}$  as 
$e(x,x) = 0$ and for $x \neq y$ $e(x,y) = g(\{x,y\})$. $e$ is symmetric and $e(x,y)=0$ iff $x=y$, 
so $e$ is semimetric. 
From \cite{mmc} Lemma 7 there is a metric $d$ such that $d(x,y) < d(x,z)$ if and only if 
$e(x,y) < e(x,z)$
 
This finishes the proof.
\end{proof}
\begin{remark}
The extension of the relation $R$ to partial order is a special case of the 
Suzumura's Extension Theorem see e.g. \cite{sl}
\end{remark}
\section{Information-theoretic context}
To bring results of II into the information-theoretic context, for certain channel 
model we will consider the partial function $f$  defined on subspace the space $X^2$ where 
$X$ is a space of codewords of the length $n$ by
$$
f(x,y) = Pr(x|y)
$$

There is one delicate point related to Theorem 1 and \cite{mmc} that needs to be discussed.
Elements of $a,b \in U$ (as defined in the proof) are incomparable by $R$ in three cases:
\begin{enumerate}
\item if $a \cap b = \emptyset $
\item $a=\{x,y\},  b=\{x,z\}, y \neq z$ but either $(x,y)\notin S$ or $(x,z)\notin S$ where S is domain of $f$
\item $a=\{x,y\},  b=\{x,z\}, y \neq z$ but $f(x,y)=f(x,z)$
\end{enumerate}

In information-theoretic interpretation of $f$ second case never happens.
For the third case, construction of the $\overline{\overline{R}}$ relation introduces order between such a 
$\{x,y\}$ and $\{x,z\}$ so, when we move to the construction of $d$  
there is implication $ Pr(x | y ) > Pr(x | z ) \Rightarrow  d(x,y) < d(x,z) $ but no implication 
$ d(x,y) < d(x,z) \Rightarrow Pr(x | y ) > Pr(x | z ) $
claimed in \cite{mmc}. Metric is still matched to the channel according to slightly modified 
Definition 1 from \cite{mmc}:

\begin{definition}
Let $W:X \rightarrow X$ be a channel with input and output alphabets $X$ and let $d$ be a metric on $X$
We say that $W$ and $d$ are matched if maximum for every code $C \subset X$ and every $x \in X$  
$\argmin_{y \in X-\{x\}} d(x,y) \subset \argmax_{y \in X-\{x\}} Pr(x|y) $  
where we interpret, that $\argmax$ returns list of size at least 1, not the single element.
\end{definition}

One (not essential) modification is that we require that range of $\argmax$ and $\argmin$ exclude $x$: 
this is because \cite{mmc} makes assumption that channel is reasonable, so $Pr(x|x)>Pr(x|y)$ whenever $x \neq y$
so without this exclusion operators trivially return $\{x\}$

More subtle is weakening the requirement  $\argmin_{y \in X-\{x\}} d(x,y) = \argmax_{y \in X-\{x\}} Pr(x|y)$ 
expressed in the same context in corresponding Definition 1 in \cite{mmc}. 
This one is essential. In fact in case of the metrics built as in the Theorem 1, 
expression $\argmin_{y \in X-\{x\}} d(x,y)$ will always evaluate to a single element list.

Some of the channels that do not have matched metrics in the sense of Defintion 1 from  \cite{mmc} do have 
in the sense of Defintion 1 as stated above.
It is such in the following case: Let $X=\{a,b,c\}$ and let

$$Pr(a|a)=Pr(b|b)=Pr(c|c)=\frac{1}{2}$$
$$Pr(a|b)=Pr(a|c)=\frac{1}{4}$$
$$Pr(b|a)=Pr(c|b)=\frac{1}{6}$$
$$Pr(b|c)=Pr(c|a)=\frac{1}{3}$$

Let's also observe, that in the case whe we assume, that $Pr(x|z)\neq Pr(x|t)$ whenever $x \neq z \neq t \neq x$ 
both definitions coincide. It would be interesting to further explore this relation between definitions 
in context of some perturbation argument, where we modify slighly the channel to assure condition above in 
consistent manner and go to the limit.

\section{Binary Assymetric Channel has matched meterics}
Now, let's apply theorem from previous section and prove following:
\begin{theorem}
Binary Assymetric Channel has matched metric (in sense of definition 1).
\end{theorem}
\begin{proof}
Let $X$ be a space of the codewords of length $m$. According to theorem from section I, if 
for each sequence of codewords $x_{0}, x_{1}, \ldots, x_{n-1}$ we will have:
$$Pr(x_{0}| x_{n-1}) > Pr(x_{0}: x_{1})$$
$$Pr(x_{1}| x_{0}) > Pr(x_{1}|x_{2})$$ 
$$\ldots$$ 
$$Pr(x_{n-2}|x_{n-3}) > Pr(x_{n-2}|x_{n-1})$$
implies: 
$$Pr(x_{n-1}|x_{n-2}) < P(x_{n-1}|x_{0})$$ 
there is a metric with required propery.
Le's assume that this is not true, so there is a sequence $x_{i}$ which satisfies 
all the inequalities from the premise but for which  
$$Pr(x_{n-1}|x_{n-2}) \geq P(x_{n-1}|x_{0})$$

Let $x_{i}(j)$ for $i \in \{0,\ldots, n-1\}$, $j \in \{0,\ldots, m-1\}$ be the $j-th$ symbol in $i-th$ codeword.
Probability of reception of the symbol codeword $x_{i}$ when $x_{(i+k) \bmod n}$, where $k \in \{1, -1\}$, was sent is then:
$$
Pr(x_{i}|x_{(i+k) \bmod n}) = \prod_{j=0}^{m-1}Pr(x_{i}(j)|x_{(i+k) \bmod m}(j))
$$
So inequalities from the premise can be written as
\begin{multline*}
\prod_{j=0}^{m-1}Pr(x_{i}(j)|x_{(i-1) \bmod m}(j))\times\\
\prod_{j=0}^{m-1}Pr(x_{i}(j)|x_{(i+1) \bmod m}(j))^{-1} > 1
\end{multline*}
for $i \in \{0,\ldots, n-2\}$
For convenience, let's move to logarithms, so we have for the same $i$:
\begin{multline*}
\sum_{j=0}^{m-1}(log(Pr(x_{i}(j)|x_{(i-1) \bmod m}(j)))-
\\log(Pr(x_{i}(j)|x_{(i+1) \bmod m}(j)))) > 0
\end{multline*}
and also
\begin{multline*}
\sum_{j=0}^{m-1}(log(Pr(x_{n-1}(j)|x_{n-2}(j)))-\\
log(Pr(x_{n-1}(j)|x_{0}(j)))) > 0
\end{multline*}
summing over $i \in \{0,\ldots, n-1\}$ we have:
\begin{multline*}
\sum_{i=0}^{n-1}\sum_{j=0}^{m-1}(log(Pr(x_{i}(j)|x_{(i-1) \bmod m}(j)))-\\
log(Pr(x_{i}(j)|x_{(i+1) \bmod m}(j)))) > 0
\end{multline*} 
Let's change the order of summation, and we will have 
\begin{multline*}
\sum_{j=0}^{m-1}\sum_{i=0}^{n-1}(log(Pr(x_{i}(j)|x_{(i-1)\bmod m}(j)))-\\
log(Pr(x_{i}(j)|x_{(i+1)\bmod m}(j)))) > 0
\end{multline*}
We claim however, that for each $j$ the inner sum is $0$:
\begin{multline*}
s_{j} = \sum_{i=0}^{n-1}(log(Pr(x_{i}(j)|x_{(i-1) \bmod m}(j)))-\\
log(Pr(x_{i}(j)|x_{(i+1) \bmod m}(j)))) = 0
\end{multline*}
so the total sum is $0$, what leads to a contradiction.
 
To show that the claim is true this, let's represent $\{x_{i}(j)\}_{i}$ of smbols $0$s and $1$s as concatenation of sequences 
consisting of identical symbols in such a way, that sequences we concatenate contain different syblos, 
e.g. sequence $(0,0,1,1,1,0)$ will be represented as conctatenation of sequences $(0,0)$, $(1,1,1)$, $0$ and let's
call this sequences respectively $X_{0}, X_{1},\ldots,X_{k-1}$ and it's elements by the the duble indexed $x_{l,i}$ wich represents. 
$i$-th element of $l$-th sequence. Because sum we consider is cyclic, we can also 
assume without loss of generality, that first and the last sequence consists of different symbols (shifting sequence cyclically if needed). 
Let's denote by $|X_{l}|$ length of the sequence $X_{l}$ We will also fix $j$ so we will not write it. We now have:

\begin{multline*}
s_{j} = \sum_{i=0}^{n-1} [log(Pr(x_{i}|x_{((i-1)\bmod m)}))-\\
log(Pr(x_{i}|x_{((i+1)\bmod m)}))] = \\
\sum_{i=0}^{n-1} log(Pr(x_{i}|x_{((i-1)\bmod m)}))-\\
\sum_{i=0}^{n-1} log(Pr(x_{i}|x_{((i+1)\bmod m)}))=\\
\sum_{l=0}^{k-1} [log(Pr(x_{l,0}|x_{((l-1)\bmod k),|X_{((l-1) \bmod k)}|-1})+\\
\sum_{i=1}^{|X_{l}|-1}log(Pr(x_{l,i}|x_{l,(i-1)})))] -\\
\sum_{l=0}^{k-1} [\sum_{i=0}^{|X_{l}|-2}log(Pr(x_{l,i}|x_{l,(i+1)}))+\\
log(Pr(x_{l,|X_{l}|-1}|x_{((l+1)\bmod k), 0}]
\end{multline*}
Let's observe, that for each $l$:
\begin{multline*}
log(Pr(x_{l,0}|x_{((l-1)\bmod k),|X_{((l-1) \bmod k)}|-1})+\\
\sum_{i=1}^{|X_{l}|-1}log(Pr(x_{l,i}|x_{l,(i-1)}))) = \\
\sum_{i=0}^{|X_{l}|-2}log(Pr(x_{l,i}|x_{l,(i+1)}))+\\
log(Pr(x_{l,|X_{l}|-1}|x_{((l+1)\bmod k), 0}))
\end{multline*}
bcause number of terms under the sign $\sum$ is the same and these terms are equal (since 
concatenated subsequences consisted of the same symbol. The terms outside the sum are equal, since
neighbour groups consists of different simbols, so on both sides their will be either $Pr(0|1)$ or $Pr(1|0)$  
This, together with contradiction pointed out earlier completes the proof. 
\end{proof}


%

\appendices
\section{Lemma from \cite{mmc}}
Following repeats the Lemma 7 from \cite{mmc} and it's proof
\begin{lemma}
Let $X$ be a finite set and $e:X^{2}\rightarrow \mathbb{R}$ be a semimetric, i.e.
satisfy following conditions:
$e(x,y) \geq 0$ for all $x,y \in X$
$e(x,y) = 0$ if and only if $x=y$
$e(x,y) = e(y,x)$ for all $x, y \in X$
then, there is a metric $d$ on $X$ such that $d(x,y) < d(x,z)$ 
if and only if $e(x,y) < e(x,z)$
\end{lemma}
\begin{proof}
Let $m=min\{e(x,y): x, y \in X, x\neq y\} > 0$ 
Let $M=max\{e(x,y): x, y \in X, x\neq y\}$
Let $\delta$ satisfying $0 <\delta <\frac{1}{3}$ be some number. 
Let $f$ be a strictly increasing bijective function 
$f:[m, M]\rightarrow[1-\delta, 1+\delta]$.
We define $d:X^2\rightarrow\mathbb{R}^{+}$ following manner:
$$
d_{x,y}=\begin{cases}
0& \text{if $x=y$ is odd},\\
f(e(x,y))& \text{otherwise}.
\end{cases}
$$
Symmetry, nonnegativity and accordance of inequalities between $e$ and $d$ 
is immediate consequence of the definition of $d$. It  also satisfies the 
triangle inequalirty since
\begin{multline*}
d(x,y)+d(y,z) \geq 2(1-\delta) > 2(1-\frac{1}{3})=\\
\frac{4}{3}>1+\delta \geq d(x,z)
\end{multline*}
so $d$ is a metric what finishes the proof of the lemma
\end{proof}


\ifCLASSOPTIONcaptionsoff
  \newpage
\fi

\end{document}